\documentclass[11pt]{article}

\usepackage{amsmath}
\usepackage{amsthm}
\usepackage{amssymb}
\usepackage{graphicx}
\usepackage{mathrsfs}

\usepackage{algorithmic}
\usepackage{multirow}

\usepackage{dsfont} 

\usepackage{tikz}
\usetikzlibrary{matrix}
\usetikzlibrary{backgrounds}
\usetikzlibrary{positioning}
\usetikzlibrary{calc}
\usetikzlibrary{arrows}
\usetikzlibrary{shapes,patterns}
\usetikzlibrary{shapes.arrows}
\usetikzlibrary{decorations,decorations.pathmorphing}
\usetikzlibrary{decorations.pathreplacing}
\usetikzlibrary{decorations.shapes}
\usetikzlibrary{decorations.text}

\usepackage[top=1in, bottom=1in, left=1in, right=1in]{geometry}

\usepackage{setspace}

\newcommand{\E}{\mathbb{E}}

\newcommand{\indicator}[1]{\mathds{1}\{ #1 \}}

\newcommand{\A}{\mathbf{A}}

\newcommand{\bd}{\mathbf{d}}

\newcommand{\T}{\mathbf{t}}
\newcommand{\Z}{\mathbf{Z}}

\newtheorem{defn}{Definition}
\newtheorem*{defn*}{Definition}
\newtheorem{prop}{Proposition}

\newtheorem{assumption}{Assumption}

\usepackage{natbib}
\title{Identification of homophily and preferential recruitment \\ in respondent-driven sampling}

\author{Forrest W. Crawford$^{1}$, Peter M. Aronow$^{1,2}$, Li Zeng$^{1}$, and Jianghong Li$^{3}$ \\[1em]
  \normalsize 
  1. Department of Biostatistics, Yale School of Public Health\\
   \normalsize 
 2. Department of Political Science, Yale University\\
  \normalsize 
3. Institute for Community Research 
} 

\date{\today}

\begin{document}
\maketitle


\begin{abstract}
  \noindent Respondent-driven sampling (RDS) is a link-tracing procedure for surveying hidden or hard-to-reach populations in which subjects recruit other subjects via their social network.  There is significant research interest in detecting clustering or dependence of epidemiological traits in networks, but researchers disagree about whether data from RDS studies can reveal it.
  Two distinct mechanisms account for dependence in traits of recruiters and recruitees in an RDS study: homophily, the tendency for individuals to share social ties with others exhibiting similar characteristics, and preferential recruitment, in which recruiters do not recruit uniformly at random from their available alters.  
  The different effects of network homophily and preferential recruitment in RDS studies have been a source of confusion in methodological research on RDS, and in empirical studies of the social context of health risk in hidden populations.  
In this paper, we give rigorous definitions of homophily and preferential recruitment and show that neither can be measured precisely in general RDS studies.  We derive nonparametric identification regions for homophily and preferential recruitment and show that these parameters are not point identified unless the network takes a degenerate form.  The results indicate that claims of homophily or recruitment bias measured from empirical RDS studies may not be credible.  We apply our identification results to a study involving both a network census and RDS on a population of injection drug users in Hartford, CT.  
\\[1em]
\textbf{Keywords:} 
hidden population, 
link-tracing, 
network sampling,
nonparametric bounds,
stochastic optimization,
social network
\end{abstract}



\section{Introduction}

Epidemiological research on the social context of health outcomes depends on researchers' ability to observe features of the social network connecting members of the target population.  In particular, many research projects seek to determine whether epidemiological traits (e.g. disease status or risk behaviors) cluster in the population social network.  But epidemiological studies of stigmatized or criminalized populations such as drug users, men who have sex with men, or sex workers can be challenging because potential subjects may be unwilling to participate in surveys or intervention campaigns because they fear exposure, persecution, or even prosecution.  Respondent-driven sampling (RDS) is a common procedure for recruiting members of hidden or hard-to-reach populations \citep{Broadhead1998Harnessing,Heckathorn1997Respondent}.  Starting with a set of initial participants called ``seeds'', subjects are interviewed and given a small number of coupons they use to recruit other members of the study population. Participants recruit others by giving them a coupon bearing a unique code and information about how to participate in the study.  Each subject receives a reward for being interviewed and another for every new subject they recruit.  

Most methodological research on RDS assumes the existence of a social network connecting members of the target population, where recruitments take place across edges in that network \citep{Salganik2004Sampling,Volz2008Probability,Gile2011Improved,Crawford2014Graphical,rohe2015network}.  For privacy reasons, subjects in an RDS study typically do not provide identifying information about their alters in the target population network.  Instead, researchers measure respondents' degree in the target population social network. Since RDS only reveals links between recruiter and recruitee, many edges in the network of respondents remain unobserved. The privacy protections afforded to subjects in an RDS survey may encourage participation, but unobserved edges impose limitations on what researchers can learn about the underlying network.

Since the recruitment process is network-based, the traits of recruiter and recruitee may not be independent \citep{Gile2010Respondent,Tomas2011Effect}.  Two mechanisms account for this dependence. First, \emph{homophily} is the tendency for people to exhibit social ties with others who share their traits \citep{McPherson2001Birds}.  Second, recruiters in RDS choose new recruits from among their network neighbors, who may share similar traits or behaviors.  \emph{Preferential recruitment} of a certain type of person, conditional on existing social ties, can make RDS recruitment chains appear more homogeneous, even in the absence of homophily in the network.  While homophily is a property of the target population social network, preferential recruitment is a property of the RDS recruitment process, conditional on that network.

Epidemiologists and public health researchers care about homophily and preferential recruitment in RDS studies because these forms of dependence may bias estimates of population-level quantities \citep{Gile2010Respondent,Tomas2011Effect,liu2012assessment,Rudolph2013Importance,rocha2015respondent}. \citet[][Table 1]{Gile2015Diagnostics} state that two assumptions ``required'' by the most popular estimator of the population mean \citep{Volz2008Probability} are ``homophily sufficiently weak'' and ``random referral''. Prospective remedies for these forms of dependence are different. The effects of homophily on estimators can sometimes be attenuated by choosing seeds in diverse populations.  Preferential recruitment is less easily diminished because this form of selection bias is controlled by subjects in the RDS study.  Epidemiologists are also interested in homophily as a measure of clustering of traits in the network.  The dynamics of infectious disease spread in populations may depend on the topological properties and traits of individuals in the epidemiological contact network \citep{salathe2010dynamics,volz2011effects}; since RDS is a network-based sampling method, it may reveal features of this contact network. For example, \citet{Stein2014Comparison} and \citet{Stein2014Online} treat RDS recruitments as epidemiological contacts to estimate assortative mixing and homophily in the close-contact network relevant for transmission of pathogens. \citet[][page 18]{Stein2014Comparison} suggest that from RDS data ``correlations between linked individuals can be used to improve parameterisation of mathematical models used to design optimal control'' for epidemic management.  However, positive correlation in the traits of recruiter and recruitee could indicate homophily, recruitment preference, both, or neither. 

Unfortunately, researchers do not agree on the definitions of homophily and preferential recruitment in RDS studies. \citet{white2015strengthening} observe that the term \emph{homophily} has ``inconsistent usage in the RDS community. Sometimes it is used to refer to the tendency for sample recruitments to occur between participants in the same social category and sometimes to refer to the tendency for relationships in the target population to occur between participants in the same social category''.  For example, \citet[][page 388]{Ramirez2005Networks} define homophily as ``a tendency toward in-group recruitment''.  \citet[][page 751]{Abramovitz2009Using} write that ``[d]ifferential recruitment patterns are usually the result of individuals' tendencies to associate with other individuals who are similar to them, also known as homophily''.  \citet[][page 307]{Uuskula2010Evaluating} define homophily as ``the extent to which recruiters are likely to recruit individuals similar to themselves''.  \citet[][page 2326]{Rudolph2014Evaluating} define preferential recruitment in network terms: ``[d]ifferential recruitment based on the outcome of interest may occur when (1) the outcome clusters by network or (2) network members cluster in space and the outcome is spatially clustered'', which seems to mirror the definition of homophily.  Finally, \citet{Fisher2014Stickiness} invent the term ``stickiness'', the tendency of recruitment chains to become stuck within a group of subjects with similar traits.  

\citet[][page 911]{Tomas2011Effect} argue that estimating homophily and preferential recruitment can be challenging in empirical RDS studies: ``it is not always possible to distinguish from the sample if differential recruitment exists, because its effect on the resulting sampling chain is similar to that of homophily''. However, many authors have claimed to measure homophily in the target population social network \citep{Gwadz2011Effect,Simpson2014Birds,rudolph2011subpopulations,wejnert2012estimating,Rudolph2014Evaluating}, and others have reported evidence of preferential recruitment in the RDS recruitment chain \citep{Iguchi2009Simultaneous,Yamanis2013Empirical,Young2014Spatial}.  Two software tools for analysis of RDS data produce estimates of homophily or preferential recruitment: RDSAT \citep{Volz2012Respondent} and RDS Analyst \citep{Handcock2013RDS}, but the estimators used by these programs are not documented, and their statistical properties have not been described in the peer-reviewed literature.

In this paper, we adapt ideas from the domain of partial identification \citep{Manski2003Partial} to the network setting \citep{depaula2014identification,graham2015methods}.  We compute nonparametric graph-theoretic bounds for homophily and preferential recruitment under minimal assumptions about the underlying network and recruitment process.  We first give rigorous definitions of homophily and preferential recruitment, and show that these quantities are not point identified unless the recruitment tree is identical to the underlying subgraph.  We describe a stochastic optimization algorithm for finding these bounds, and give conditions for its convergence.  To illustrate the bounds, we analyze data from a unique RDS survey of people who inject drugs (PWID) in Hartford, Connecticut in which the subgraph of respondents and their network alters is known with near certainty.  We compare the point estimates of homophily and preferential recruitment obtained by using the full subgraph information with the identification intervals computed using the RDS data alone.  


\section{Preliminaries}

\subsection{Basic assumptions}

\label{sec:assump}

We first state some basic assumptions about the population social network and the RDS recruitment process. These assumptions are implicit in the original work on statistical inference for RDS \citep{Heckathorn1997Respondent,Salganik2004Sampling,Volz2008Probability,Gile2010Respondent,Gile2011Improved}.  First, we place RDS in its proper network-theory context. 
\begin{assumption}
 The social network connecting members of the target population exists and is an undirected graph $G=(V,E)$ with no parallel edges or self-loops.
 \label{assump:net}
\end{assumption}
\noindent Members of the target population are vertices in $V$, and edges in $E$ represent social ties between individuals.  A seed is a vertex that is not recruited, but is chosen by some other mechanism, not necessarily random.  A recruiter is a vertex known to the study that has at least one coupon.  A susceptible vertex is not yet known to the study, but has at least one neighbor in $G$ that is a recruiter. 
Every vertex $i\in V$ has a binary attribute, trait, or covariate $Z_i$ that is observed only when $i$ is recruited.  We focus here on binary attributes for simplicity, but the arguments presented below can be extended to continuous covariates by introducing a similarity metric.  
\begin{assumption}
  RDS recruitments happen across edges in $G$ connecting a recruiter to a susceptible vertex.
  \label{assump:recruit}
\end{assumption}
\noindent Finally, we state a practical assumption related to the conduct of real-world RDS studies. 
\begin{assumption}
  No subject can be recruited more than once.  
  \label{assump:nomult}
\end{assumption}
\noindent While Assumption \ref{assump:nomult} is always followed in empirical RDS studies, it is ignored in idealized models of recruitment used to justify the form of traditional RDS estimators \citep{Salganik2004Sampling,Volz2008Probability}.  Since Assumption \ref{assump:nomult} is always true in practice, we will always take it as true in what follows. 


\subsection{The structure of RDS data}

We now define the network data collected by typical RDS studies.  Variants of these definitions were first given by \citet{Crawford2014Graphical}.  Under Assumptions \ref{assump:net}, \ref{assump:recruit}, and \ref{assump:nomult} the RDS recruitment path reveals a subgraph of $G$ that is observable by researchers.  
\begin{defn}[Recruitment graph]
  The directed recruitment graph is $G_R=(V_R,E_R)$, where $V_R$ is the set of $n$ sampled vertices (including seeds) and a directed edge from $i$ to $j$ indicates that $i$ recruited $j$.
\end{defn}
\noindent 
Knowledge of the edges connecting observed vertices reveals the induced subgraph of respondents.
\begin{defn}[Recruitment-induced subgraph] 
  The recruitment-induced subgraph is an undirected graph $G_S=(V_S, E_S)$, where $V_S=V_R$ consists of $n$ sampled vertices; and $\{i,j\}\in E_S$ if and only if $i\in V_S$, $j\in V_S$, and $\{i,j\}\in E$.
\end{defn}
\noindent Subjects also report the number of people they know (but not their identities) who are members of the target population.
\begin{defn}[Degree]
  The degree $d_i$ of $i\in V_R$ is the number of edges incident to $i$ in $G$. 
\end{defn}
\noindent Let $\bd_R$ and $\mathbf{t}$ be the $n\times 1$ vectors of recruited vertices' degrees and times of recruitment in the order they entered the study, and let $M$ be the set of seeds.  Label the vertices in $V_R$ $1,\ldots,n$, in the order of their recruitment in the study.  Label the remaining vertices in $V\setminus V_R$ arbitrarily with the numbers $n+1,\ldots,N$.  Furthermore, we observe a vector $\Z_R=(Z_1,\ldots,Z_n)$ of subjects' binary trait values.  Researchers conducting an RDS study only observe $G_R$, $\mathbf{d}_R$, $\mathbf{t}_R$, and $\Z_R$.  

One more definition will assist us in defining sample measures of homophily and preferential recruitment.  Let $U=\{u\notin V_R:\ \exists\ v\in V_R \text{ with } \{v,u\}\in E\}$ be the set of unsampled vertices connected by at least one edge to a sampled vertex in $V_R$ at the end of the study. Let $E_U = \{ \{v,u\}:\ v\in V_R,\  u\in U\text{, and } \{v,u\}\in E\}$ be the set of edges connecting vertices in $U$ to sampled vertices in $V_R$.  
\begin{defn}[Augmented recruitment-induced subgraph]
  The augmented recruitment-induced subgraph is an undirected graph $G_{SU}=(V_{SU},E_{SU})$, where $V_{SU} = V_S \cup U$ and $E_{SU}=E_S \cup E_U$.  
  \label{defn:augmented}
\end{defn}
\noindent Note that $G_{SU}$ does not contain edges between vertices in $U$, and contains no vertices that are not connected to a vertex in $V_R$. Let $\Z_{SU}$ be the set of traits of all vertices in $V_{SU}$.  Figure \ref{fig:augmented} shows an example of a population graph $G$, the recruitment graph $G_R$, the recruitment-induced subgraph $G_S$, and the augmented recruitment-induced subgraph $G_{SU}$.

\begin{figure}
  \centering
  \includegraphics[width=0.8\textwidth]{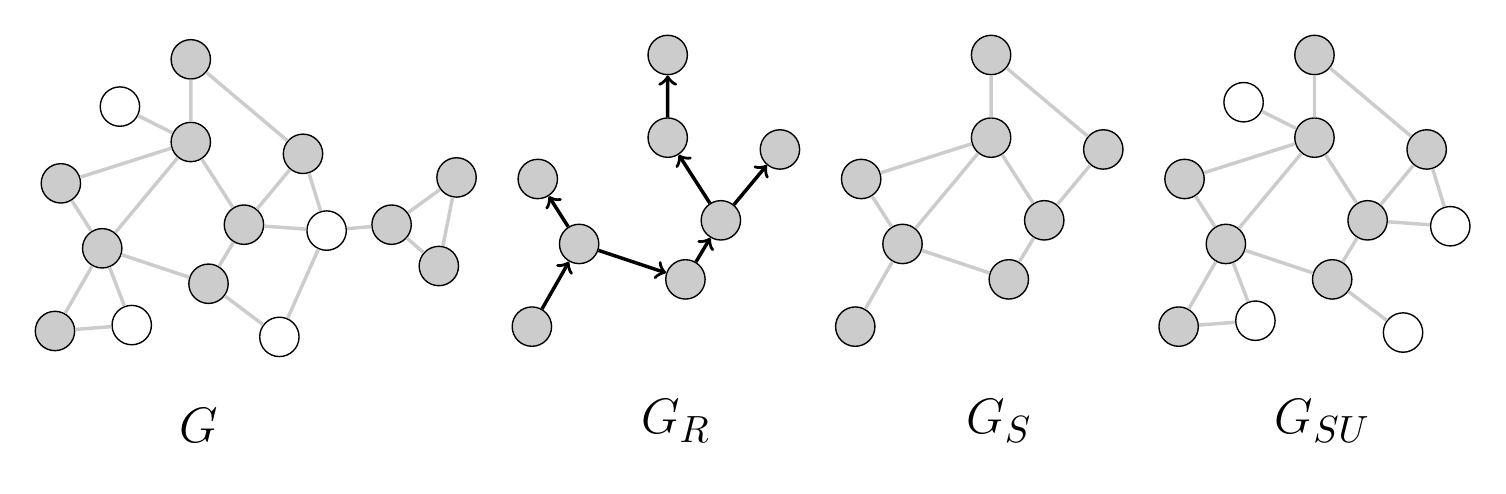}
  \caption{Illustration of the population graph $G$, the recruitment graph $G_R$, the recruitment-induced subgraph $G_S$, and the augmented recruitment-induced subgraph $G_{SU}$. Only $G_R$ and the degrees of recruited vertices are revealed to researchers conducting an RDS study. However, the degrees of recruited vertices place topological constraints on the space of augmented recruitment-induced subgraphs $G_{SU}$. }
 \label{fig:augmented}
\end{figure}


\section{Definitions and inferential targets}

Suppose $G=(V,E)$ is the population graph.  Consider an RDS sample of size $n$ with recruitment graph $G_R=(V_R,E_R)$, degrees $\bd_R$, and traits $\Z_R$.  Let $G_{SU}=(V_{SU},E_{SU})$ be the augmented subgraph for this sample, with traits $\Z_{SU}$.  The observed traits $\Z_R$ are a subset of $\Z_{SU}$.  Our inferential targets are the subgraph homophily and the subgraph preferential recruitment, defined formally below.  These parameters are data-adaptive \citep{vanderlaan2013statistical,balzer2015targeted}: they are properties of the network and trait values of vertices proximal to the RDS sample.


\subsection{Homophily}

Let $\A=\{A_{ij}\}$ be the $N\times N$ adjacency matrix of the population graph $G$ and let $A_{ij}$ be a binary variable indicating presence or absence of an undirected edge between $i$ and $j$. 
\begin{defn}[Subgraph homophily]
  The subgraph homophily is the correlation between $A_{ij}$ and the indicator $\indicator{Z_i=Z_j}$, conditional on $i\in V_R$ and $j\in V_{SU}$, 
  \begin{equation}
     h(G,G_R,\Z) = \rho(A_{ij}, \indicator{Z_i=Z_j}|i\in V_R,j\in V_{SU}) 
  \label{eq:homophily}
  \end{equation}
  where $\rho(\cdot)$ denotes the Pearson correlation over every $i\in V_R$ and $j\in V_{SU}$.
\label{defn:homophily}
\end{defn}
\noindent Inference about homophily in the population graph is identical to inference about homophily in the augmented subgraph. 
\begin{prop}
  $h(G,G_R,\Z) = h(G_{SU},G_R,\Z_{SU})$.
 \label{prop:homophily}
\end{prop}
\noindent The proof follows directly from observation that $i\in V_R$ implies $i\in V_{SU}$.  To ease notation, we will often use $h$ to refer to $h(G,G_R,\Z)$.  To compute $h(G,G_R,\Z)$, suppose $(G_{SU},\Z_{SU})$ is known and let $\A_{SU}$ be the adjacency matrix of $G_{SU}$.  Then by Proposition \ref{prop:homophily}, the subgraph homophily defined in \eqref{eq:homophily} can be computed as 
\begin{equation}
  h(G_{SU},G_R,\Z_{SU}) =  \frac{\sum_{i\in V_R}\ \sum_{j\in V_{SU},j\neq i} (A_{ij} - \bar{A}_{SU})(\indicator{Z_i=Z_j} - \bar{Z}_{SU})}{\left(\binom{|V_R|}{2} + |V_R||U| \right) \sigma(\A_{SU})\ \sigma(\Z_{SU})}
  \label{eq:homophily2}
\end{equation}
where $\binom{|V_R|}{2} + |V_R||U|$ is the number of potential edges in $G_{SU}$, and
\[ \bar{A}_{SU} = \frac{ |E_{SU}|}{\binom{|V_R|}{2} + |V_R||U|},  \]
\[ \bar{Z}_{SU} = \frac{\sum_{i\in V_R}\ \sum_{j\in V_{SU},j\neq i} \indicator{Z_i=Z_j} }{ \binom{|V_R|}{2} + |V_R||U|},  \]
\[ \sigma(\A_{SU}) = \sqrt{\frac{ \sum_{i\in V_R}\ \sum_{j\in V_{SU},j\neq i} (A_{ij} - \bar{A}_{SU})^2  }{\binom{|V_R|}{2} + |V_R||U|} }, \]
and 
\[ \sigma(\Z_{SU}) = \sqrt{\frac{\sum_{i\in V_R}\ \sum_{j\in V_{SU},j\neq i} (\indicator{Z_i=Z_j} - \bar{Z}_{SU})^2 }{\binom{|V_R|}{2} + |V_R||U|} }. \]


\subsection{Preferential recruitment}

Let $S_i(t)$ be the set of susceptible neighbors of $i\in V_R$ just before time $t$ (the set-valued function $S_i(t)$ is left-continuous).  Recall that $\mathbf{t}_R=(t_1,\ldots,t_n)$ is the vector of times of recruitments, so $S_i(t_j)$ is the set of susceptibles connected to $i$ just before recruitment of $j$.  Let $r_j$ be the recruiter of the sampled vertex $j\in V_S$ ($r_j$ is undefined when $j$ is a seed).  Let $\mathbf{t}=(t_1,\ldots,t_N)$ be the ordered vector of recruitment times, where the first $n$ times $t_1,\ldots,t_n$ are finite, and the remaining times $t_{n+1}=\cdots=t_N=\infty$.  Let $i\in V$ be a vertex, let $S\subseteq V$ where $i\notin S$, and let
$\text{same}(i,S) = |\{k\in S:\ Z_i=Z_k\}|$
be the number of vertices in $S$ that have the same trait value as $i$.  Then $\text{same}(i,S_i(t_j))$ is the number of same-type susceptible vertices connected to a recruiter $i$ at the time of the $j$th recruitment.  For $j\notin M$, let $Y_j = \indicator{Z_j=Z_{r_j}}$ be the indicator that recruited vertex $j$ has the same trait as its recruiter.  Let $\Z(S)$ be the set of trait values indexed by the set $S$.  We first define proportional, or unbiased, recruitment. 
\begin{defn}[Proportional recruitment]
  Recruitment of $j$ by $r_j$ is proportional if $\Pr(Y_j|j\in S_{r_j}(t_j), \Z(S_{r_j}(t_j))) = \text{same}(r_j,S_{r_j}(t_j))/|S_{r_j}(t_j)|$. 
 \label{defn:proprec}
\end{defn}
\noindent Under proportional recruitment, the probability of a recruiter recruiting a susceptible neighbor with the same trait value is proportional to the number of its susceptible neighbors with the same trait.  In other words, recruitment is uniformly at random among susceptible neighbors.  Let $Y_j^\text{prop}$ be the outcome of a recruitment event in which the recruiter $r_j$ obeys proportional recruitment as in Definition \ref{defn:proprec}. 
\begin{defn}[Subgraph preferential recruitment]
  The subgraph preferential recruitment is the average deviation from proportional recruitment, given knowledge of $G$, $G_R$, $\mathbf{t}$, and $\Z$,  
\begin{equation}
  p(G,G_R,\mathbf{t},\Z) = \E[Y_j|j\in S_i(t_j),\Z(S_i(t_j))] - \E[Y_j^\text{prop}|j\in S_i(t_j),\Z(S_i(t_j))] 
  \label{eq:prefrec}
\end{equation}
where the expectation is over recruitments, and $i$ is the recruiter. 
\label{defn:prefrec}
\end{defn}
\noindent As before, inference about preferential recruitment in the population graph is identical to inference about preferential recruitment in the augmented subgraph. 
\begin{prop}
$p(G,G_R,\T,\Z) = p(G_{SU},G_R,\T_{SU},\Z_{SU})$.
\label{prop:prefrec}
\end{prop} 
\noindent The proof follows from the same reasoning as the proof Proposition \ref{prop:homophily}.  Definition \ref{defn:proprec} and Proposition \ref{prop:prefrec} allow us to write the subgraph preferential recruitment as 
\begin{equation}
  p(G_{SU},G_R,\mathbf{t}_R,\Z_{SU}) = \frac{1}{n-|M|} \sum_{j\notin M} \left( Y_j - \frac{\text{same}(r_j,S_{r_j}(t_j))}{|S_{r_j}(t_j)|} \right) 
\label{eq:prefrec2}
\end{equation}
where the sum is over non-seed recruited vertices $j$. To ease notation, we will sometimes write $p$ in place of $p(G,G_R,\T,\Z)$.


\section{Identification}

The parameters $h$ and $p$ depend on possibly unobserved edges between pairs of recruited vertices, and between recruited vertices and unrecruited vertices.  The observed recruitment subgraph $G_R$ and reported degrees $\bd_R$ place strong topological restrictions on the structure of $G_{SU}$, and hence imply restrictions on $h$ and $p$.  
\begin{defn}[Compatibility] \label{defn:compatibility}
  The pair $(G_{SU},\Z_{SU})$ is compatible with the observed data $G_R$, $\bd_R$, and $\Z_R$ if 1) the set of recruited vertices is preserved: $V_R \subseteq V_{SU}$; 2) the set of recruitment edges is preserved: $E_R \subseteq E_{SU}$; 3) the set of recruited subjects' trait values is preserved $\Z_R\subseteq \Z_{SU}$; 4) all unsampled vertices are connected to a recruited vertex: every $u\in V_{SU}$ with $u\notin V_S$ has an edge $\{v,u\}$ such that $v\in V_S$; and 5) total degree is preserved: for every $i\in V_R$, $d_i = \sum_{j\in V_{SU}} \indicator{ \{i,j\}\in E_{SU} }$.
\end{defn}
\noindent Let $\mathcal{C}(G_R,\bd_R,\Z_R)$ be the set of pairs $(G_{SU},\Z_{SU})$ compatible with the observed data in the sense of Definition \ref{defn:compatibility} (this is a finite set).  

First, we examine whether the recruitment-induced subgraph $G_S$ and augmented recruitment-induced subgraph $G_{SU}$ are revealed by the observed data in RDS.  Recall that $d_i$ is the total degree of $i$ and let $d_i^r$ be the degree of subject $i$ in the recruitment subgraph $G_R$.  
\begin{prop} 
  Suppose there exist $i\in V_R$ and $j\in V_R$ with $i\neq j$, $d_i^r<d_i$, and $d_j^r<d_j$.  Then neither $G_S$ nor $G_{SU}$ are identified.  
  \label{prop:GS}
\end{prop}
\noindent Proof is given in the Appendix. This result establishes the conditions under which statements about the population graph proximal to the sample can be made precise.  Next, we define the information about $h$ and $p$ that is revealed by the observed data.  
\begin{defn}[Identification region] \label{defn:identification}
  The identification regions for $h$ and $p$ are given by the smallest intervals that contain $h(G_{SU},G_R,\Z_{SU})$ and $p(G_{SU},G_R,\mathbf{t}_R,\Z_{SU})$ for $(G_{SU},\Z_{SU})\in \mathcal{C}(G_R,\bd_R,\Z_R)$.
\end{defn}
\noindent When the identification region for $h$ or $p$ contains only a single point, that parameter is point identified.  Bounds for $h$ and $p$ on the set $\mathcal{C}(G_R,\bd_R,\Z_R)$, as given in Definition \ref{defn:identification}, are sharp: there is no narrower bound that contains all possible values of these parameters.  
\begin{defn}[Identification rectangle]
  The identification rectangle for $h$ and $p$ (provided these quantities are defined) is the smallest rectangle in $[-1,1]\times [-1,1]$ that contains all values of $\big(h(G_{SU},G_R,\Z_{SU}), p(G_{SU},G_R,\mathbf{t}_R,\Z_{SU})\big)$ for $(G_{SU},\Z_{SU})\in \mathcal{C}(G_R,\bd_R,\Z_R)$.
  \label{defn:idrect}
\end{defn}
\noindent The identification rectangle is obtained by taking the Cartesian product of the identification regions for $h$ and $p$.  Finally, we provide sufficient conditions for the identification regions for $h$ and $p$ to contain more than one point. 
\begin{prop}
  Suppose there exist two vertices $i\in V_R$ and $j\in V_R$ such that $d_i>d_i^r$ $d_j>d_j^r$, and $Z_i\neq Z_j$.  Then $h$ is not point identified.
 \label{prop:homophilyid}
\end{prop}
\begin{prop}
 Suppose there exists a vertex $i\in V_R$ who recruited at least one other vertex $j\in V_R$, $j\neq i$, and $d_i>d_i^r$.  Then $p$ is not point identified.
 \label{prop:prefrecid}
\end{prop}
\noindent Proof is given in the Appendix.  

In practice, point identification of both subgraph homophily and subgraph preferential recruitment can only be achieved if the recruitment graph $G_R$ is nearly identical to the augmented recruitment-induced subgraph $G_{SU}$.  By Assumption \ref{assump:nomult}, the recruitment subgraph $G_R$ is acyclic, so $G_R=G_{SU}$ means that the population network proximal to recruited vertices is a tree, a situation that seems unlikely to occur in a real-world social network.  Furthermore, Propositions \ref{prop:homophilyid} and \ref{prop:prefrecid} apply directly to the case where all vertices in the population have been sampled, and we have $V_R=V$: if pendant edges remain, then $h$ and $p$ may not be point identified.


\section{Stochastic optimization for extrema of $h$ and $p$}

\label{sec:simanneal}

Unfortunately there are no general closed-form expressions for the extrema of $h$ and $p$ on $\mathcal{C}(G_R,\bd_R,\Z_R)$. The space of compatible subgraphs described by Definition \ref{defn:compatibility} can be very large, but straightforward optimization techniques permit finding these bounds quickly.  In this Section we introduce a stochastic optimization algorithm for finding the global optimum of an arbitrary function $J$ of $h$ and $p$, based on simulated annealing \citep{kirkpatrick1983optimization,cerny1985thermodynamical,Hajek1988Cooling,bertsimas1993simulated}.  The approach is similar to a quadratic programming framework introduced by \citet{depaula2014identification} for finding the identification set for certain functionals of graphs and vertex attributes.

Let $J:[-1,1]^2 \to \mathbb{R}$ be a function taking arguments $h(G_{SU},G_R,\Z_{SU})$ and $p(G_{SU},G_R,\mathbf{t}_R,\Z_{SU})$ for $(G_{SU},\Z_{SU})\in \mathcal{C}(G_R,\bd_R,\Z_R)$.  We choose this function, abbreviated $J(h,p)$, so that a desired feature of $\mathcal{C}(G_R,\bd_R,\Z_R)$ coincides with the maximum of $J$.  For example, the maximum of the function 
\[ J(h,p) = \frac{1}{1+\epsilon+ h} \]
on $\mathcal{C}(G_R,\bd_R,\Z_R)$ where $\epsilon>0$, coincides with the lower identification bound of $h$.  For concreteness in what follows, we will assume $J(h,p)$ has this form; similar definitions can be formulated individually to find the maximum of $h$, and the minimum and maximum of $p$. 

For $T>0$, define the objective function $\pi(h,p) \propto \exp[ J(h,p) / T ]$.  Our goal is to find $(G_{SU},\Z_{SU})\in \mathcal{C}(G_R,\bd_R,\Z_R)$ such that $\pi\big(h(G_{SU},G_R,\Z_{SU}), p(G_{SU},G_R,\mathbf{t}_R,\Z_{SU})\big)$ is maximized.  Let 
\[ K\big( (G_{SU},\Z_{SU}), (G_{SU}^*,\Z_{SU}^*)\big) \] 
be a transition kernel that describes the probability of moving from a state $(G_{SU},\Z_{SU})\in \mathcal{C}(G_R,\bd_R,\Z_R)$ to another state $(G_{SU}^*,\Z_{SU}^*)\in \mathcal{C}(G_R,\bd_R,\Z_R)$.  Let $T_t$ be a positive non-decreasing sequence indexed by $t$, with $\lim_{t\to\infty} T_t = 0$.  We construct an inhomogeneous Markov chain on $\mathcal{C}(G_R,\bd_R,\Z_R)$.   At step $t$, where the current state is $(G_{SU},\Z_{SU})$, we accept the proposed state $(G_{SU}^*,\Z_{SU}^*) \sim K\big((G_{SU},\Z_{SU}),\cdot\big)$ with probability 
\[ \rho_t = \min\left\{1, \exp\left[ \frac{J\big(h(G_{SU}^*,G_R,\Z_{SU}^*),p(G_{SU}^*,G_R,\mathbf{t}_R,\Z_{SU}^*)\big) - J\big(h(G_{SU},G_R,\Z_{SU}),p(G_{SU},G_R,\mathbf{t}_R,\Z_{SU})\big) }{T_t} \right] \right\} . \]
The proposal function is described formally in the Appendix. 

As $T_t\to 0$, the samples $(G_{SU},\Z_{SU})_t$ become more concentrated around local maxima of $\pi$.  Convergence of the sequence $(G_{SU},Z_{SU})_t$ to a global optimum depends on its ability to escape local maxima of $J$. The sequence $T_t$, called the ``cooling schedule'', controls the rate of convergence.  Let $\mathcal{M}$ denote the set of $(G_{SU},\Z_{SU})\in \mathcal{C}(G_R,\bd_R,\Z_R)$ for which $J\big( h(G_{SU},G_R,\Z_{SU}), p(G_{SU},G_R,\mathbf{t}_R,\Z_{SU})\big)$ is equal to the global maximum.  
Careful choice of $T_t$ ensures that the sequence of samples converges in probability to an element of $\mathcal{M}$. 
\begin{prop}
  Let the cooling schedule be given by $T_t = \big(\epsilon\log(t)\big)^{-1}$ where $\epsilon>0$ is a constant.  Then $\lim_{t\to\infty} \Pr\big( (G_{SU},\Z_{SU})_t \in \mathcal{M} \big) = 1$ .
  \label{prop:simanneal}
\end{prop}
\noindent The proof, which is an application of the result by \citet{Hajek1988Cooling}, is given in the Appendix. 

The optimization routine described here and in the Appendix is constructive: it returns the (possibly not unique) pair $(G_{SU},\Z_{SU})\in \mathcal{C}(G_R,\bd_R,\Z_R)$ that maximizes $\pi$.  Figure \ref{fig:example} shows a simple example RDS dataset $(G_R,\Z_R)$, the identification rectangle for $h$ and $p$, and compatible elements $(G_{SU},\Z_{SU})$ that achieve these bounds.  At top left is the recruitment subgraph $G_R$ with vertices shaded according to their type, and the pendant edges implied by each vertex's degree.  At top right is the identification rectangle whose boundaries are the extrema of $h$ and $p$. At bottom are the compatible subgraphs that achieve these extrema. The initial pair $(G_{SU},\Z_{SU})$ is chosen by randomly connecting pendant edges of recruited vertices to other recruited vertices, then connecting any remaining pendant edges to unique unrecruited vertices, with randomly assigned trait value.

\begin{figure}
  \centering
  \includegraphics[width=0.8\textwidth]{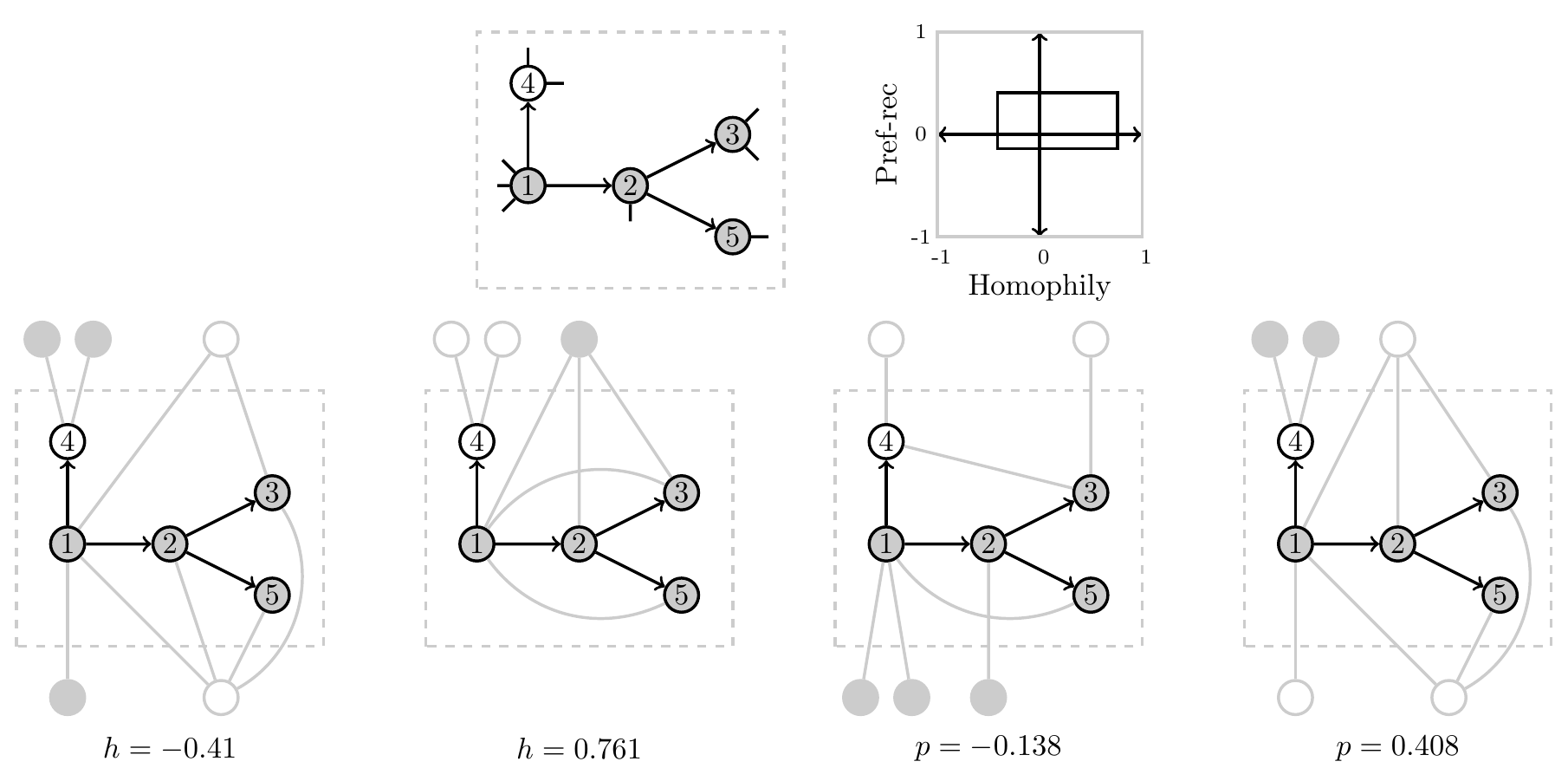}
  \caption{Illustration of extrema for $h$ and $p$ with corresponding augmented subgraphs $G_{SU}$ and $\Z_{SU}$. At top left, $G_R$ is shown with arrows indicating recruitments, and pendant edges implied by vertex degrees.  Vertices are labeled by the order of recruitment, and shaded according to their binary type. At top right, the joint homophily/preferential recruitment space is shown, and the identification rectangle for $(h,p)$ is shown.  Below, the augmented subgraphs corresponding to the minimum and maximum values for homophily and preferential recruitment are shown. Some extremal subgraphs are not unique.}
  \label{fig:example}
\end{figure}


\section{Application: injection drug users in Hartford, CT}

\label{sec:rdsnet}

\subsection{Study overview}

We now apply the ideas developed above to an extraordinary RDS dataset in which the augmented subgraph $G_{SU}$ of an RDS sample $G_R$ is known with near certainty.  In the RDS-net study, researchers conducted an RDS survey of $n=|V_R|=527$ injection drug users from $|M|=6$ seeds in Hartford, Connecticut.  Researchers simultaneously performed a census of the augmented recruitment-induced subgraph, consisting of $|V_{SU}|=2626$ unique injection drug users.  The primary purpose of RDS-net was to assess the dynamics of recruitment in a high-risk population of drug users; some details of this study design have been reported previously \citep{Mosher2015Qualitative}. Subjects were given \$25 for being interviewed, \$10 for recruiting another eligible subject (up to a maximum of three) and \$30 for completing a 2-month follow-up interview.  Subjects were required to be at least 18 years old, reside in the Hartford area, and to report injecting illicit drugs in the last 30 days.  The dates and times of recruitment were recorded for all sampled subjects. The study was approved by the Institute for Community Research institutional review board, and informed consent was obtained from all subjects. 

This study differs from typical RDS surveys because in addition to reporting their network degree, respondents also enumerated (nominated) their network alters -- other people eligible for the study whom they knew by name and could possibly recruit. Unsampled injection drug users nominated by more than one participant were matched using identifying characteristics including name (including aliases), photograph, multiple addresses, phone numbers, locations frequented, and social network links \citep{li2012social}. Comprehensive locator data, multiple data sources, and field observation notes were used to facilitate the matching process.  Outreach workers with expertise in the local injection drug-using community made a final assignment of unique subject indentifiers for nominated subjects, and links between subjects. This matching process revealed connections to unrecruited subjects along which no recruitment event took place, and resolved uniquely any unrecruited subjects nominated by more than one recruited subject \citep{weeks2002social}. The resulting ``nomination'' network is the augmented subgraph $G_{SU}$ described in Definition \ref{defn:augmented}. The usual RDS recruitment graph $G_R$ is a known subgraph of this nomination network. Figure \ref{fig:rdsnetdata} shows the nomination network $G_{SU}$ with the recruitment graph $G_R$ overlaid.  A participant's degree was defined as the sum of the number of people they nominated and the number of additional people recruited but not initially nominated. During the follow-up interview, 119 recruited subjects reported that they had given a coupon to someone other than the person who eventually returned one of their coupons. We therefore defined the recruitment graph $G_R$ as the network of coupon redemptions; we defined network alters of a recruiter in $G_{SU}$ as the \emph{union} of their nominees, and the individuals recruited using their coupons.

Demographic and trait data relevant to drug use were collected about each recruited subject.  Each recruited subject also reported the traits of their nominees.  Nominees who were never personally interviewed were assigned trait information as follows: if their nominating alters agreed on their trait value, that value was assigned to them.  If there was disagreement, the modal value was assigned. When trait information for a recruited subject or an unrecruited alter was absent or contradictory, it was treated as missing. 

\begin{figure}
  \centering
  \includegraphics[width=0.5\textwidth]{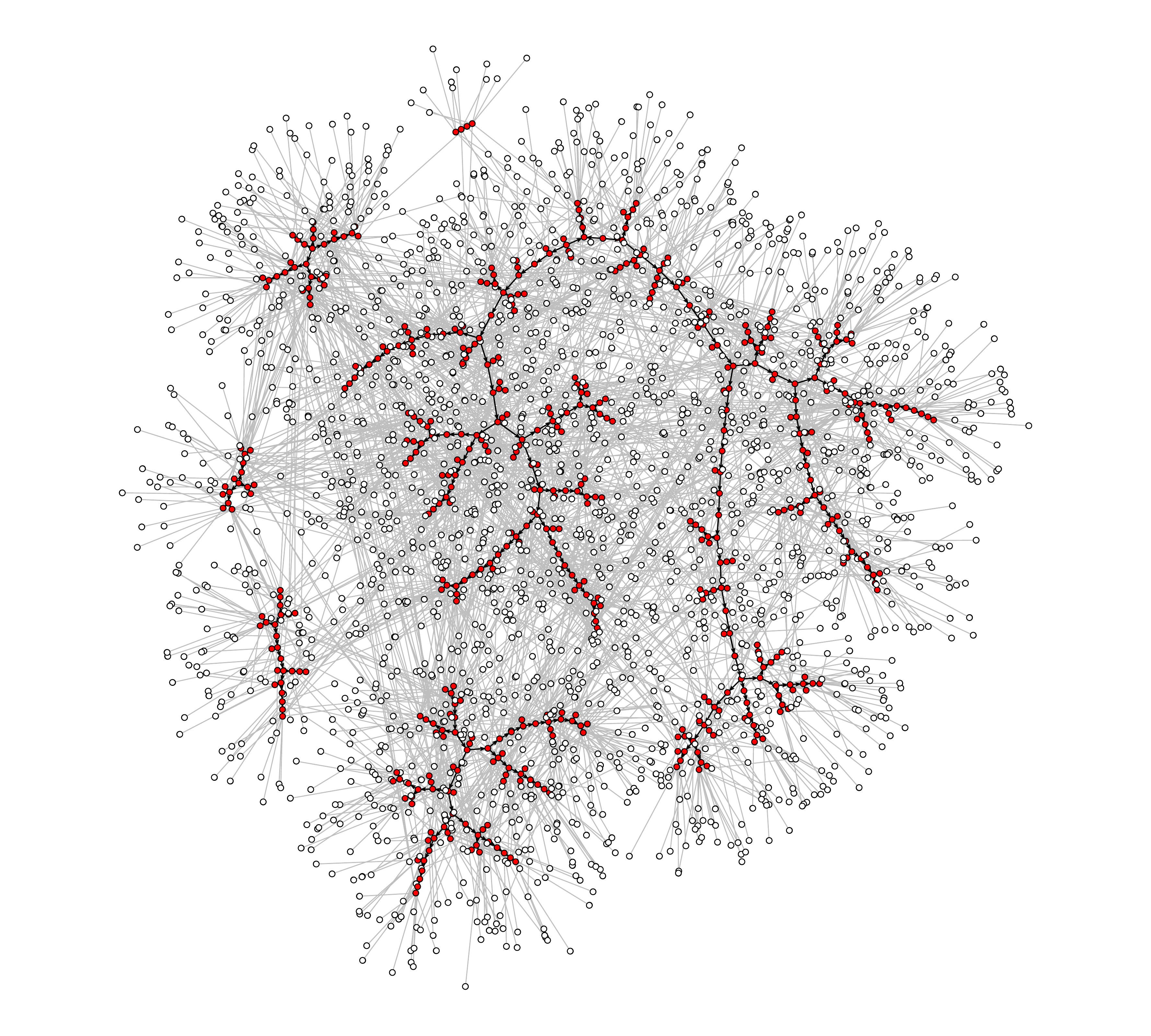}
  \caption{The nomination network and recruitment graph for injection drug users in Hartford, CT from the RDS-net study. Recruited subjects are shown in red, and recruitment edges are shown as a directed edge (arrow) from the recruiter to the recruitee.  Non-recruitment edges (linking recruited subjects to unrecruited subjects, or recruited to recruited subjects) are shown in gray. }
  \label{fig:rdsnetdata}
\end{figure}

\subsection{Descriptive results}

The RDS-net study includes $|V_{SU}|=2626$ people, of which $n=|V_R|=527$ were recruited subjects, and 2099 nominated but never-recruited subjects. There are $|E_{SU}|=3307$ edges in $G_{SU}$, of which $|E_S|=1180$ link recruited subjects to recruited subjects, and 2127 link recruited subjects to unrecruited subjects.  The mean overall degree of recruited subjects is $8.5\pm 3.8$ with maximum 26, and the mean degree of subjects in the recruitment-induced subgraph $G_S$ is $4.5\pm 2.7$ with maximum 22.  The mean number of network alters recruited by each subject is $0.99\pm 0.96$, with a maximum of 3. Non-seed recruitment was effective: while 201 people recruited no other subjects, 176 recruited one subject, 105 recruited 2 subjects, and 45 recruited three other subjects.

We selected three traits with the least missing data for analysis: gender, ``crack'' cocaine use, and homelessness.  This information is fully observed for every recruited subject, but some values are missing for nominated, but unrecruited subjects.  Table \ref{tab:rdsnettraits} provides summaries of these traits for all subjects in the study.  The gender variable contains only one missing value.  One subject reported being transgender, neither male nor female; we did not alter this value, so tests of equality $Z_i=Z_j$ for the ``gender'' trait are always false for this person. Crack and homelessness data were less complete for many unrecruited nominees.

\begin{table}
  \centering
  \begin{tabular}{lccccc}
    \hline
Trait    &&   0  &    1 & Other & Missing \\
\hline
Gender   && 720  & 1904 &     1 & 1   \\
Crack    && 1320 & 1173 &       & 133 \\
Homeless && 1222 &  853 &       & 551 \\
\hline
  \end{tabular}
  \caption{Traits used for analysis in the RDS-net study. For gender, 0 indicates female, 1 indicates male, and ``other'' refers to one subject who reported being transgender, neither male nor female. Crack use and homelessness were assigned value 1 for positive status and 0 otherwise.  All recruited subjects' values of these traits are known, but not all trait values were observed for some unrecruited nominated subjects.  Most subjects' gender identity was observed, but substantial numbers of subjects lacked crack use and homelessness data.}
  \label{tab:rdsnettraits}
\end{table}

\subsection{Homophily and preferential recruitment}

We find estimates of homophily and preferential recruitment for each trait under two scenarios.  In the first, we omit vertices in $V_{SU}$ whose trait is missing (all recruited subjects' trait values are fully observed), any edges incident to these vertices, and the corresponding elements of $\Z_{SU}$.  Then the data are given by $(G_{SU},G_R,\mathbf{t},\Z_{SU})$, so we calculate the parameters $h$ and $p$, which are point identified in the absence of missing data.  In the second scenario, we compute bounds for $h$ and $p$ using only the data observed in the RDS portion of the study, $(G_R,\bd_R,\mathbf{t}_R,\Z_R)$.  This is the setting in which most researchers analyze data from RDS studies.  Starting compatible subgraphs were chosen randomly from $\mathcal{C}(G_R,\bd_R,\Z_R)$ by first connecting a random number of pendant edges belonging to recruited vertices, while avoiding parallel edges and self-loops.  Then, any remaining pendant edges were connected to unsampled vertices, whose trait was assigned value 1 with probability 1/2 and zero otherwise. We assessed convergence of the optimization routine from multiple randomly selected starting graphs; convergence was not sensitive to the starting point.

Table \ref{tab:results} shows the results, where point estimates of $h$ and $p$ are given under omission of vertices with missing data.  The intervals $I_h$ and $I_p$ give the identification bounds obtained using the observed RDS data alone.  Point estimates (where vertices with missing traits are excluded) always lie within the identification intervals.  The point estimates for homophily $h$ with respect to gender and crack use are positive, and negative for homelessness, while $p$ is positive for gender and homelessness,but negative for crack use.  Figure \ref{fig:results} shows the identification rectangles obtained by taking the Cartesian product of $I_h$ and $I_p$ in Table \ref{tab:rdsnettraits}.  All identification rectangles cover $(0,0)$.  The point estimates are given by a circle.  The four traces, corresponding to the minima and maxima of $h$ and $p$, show the paths of values $(h,p)$ taken by the optimization algorithm described in Section \ref{sec:simanneal} for finding extrema of these parameters on the set $\mathcal{C}(G_R,\bd_R,\Z_R)$.  

\begin{table}
  \centering
  \begin{tabular}{lcccccc}
    \hline
    && \multicolumn{2}{c}{Homophily} && \multicolumn{2}{c}{Preferential recruitment} \\ \cline{3-4} \cline{6-7} 
    Trait    && $h$      & $I_h$ && $p$      & $I_p$  \\
    \hline 
    Gender   && 0.00878  & (-0.0779,0.0984)  &&  0.00397 & (-0.190,0.534) \\
    Crack    && 0.00283  & (-0.0841,0.0701)  && -0.00051 & (-0.272,0.453) \\
    Homeless && -0.00011  & (-0.0823,0.0729) &&  0.04527 & (-0.221,0.504) \\
    \hline 
  \end{tabular}
  \caption{Homophily $h$ and preferential recruitment $p$ in the RDS-net study for gender, crack use, and homelessness. Point values are given first, then identification intervals $I_h$ and $I_p$ for $h$ and $p$ respectively, obtained by using only the information observed in the RDS portion of the study.  The homophily point estimate is positive for gender and crack use, and negative for homelessness.  Preferential recruitment point estimates are positive for gender and homelessness, and negative for crack.  The identification intervals $I_h$ and $I_p$ contain the point estimates of $h$ and $p$ in all cases. }
  \label{tab:results}
\end{table}

\begin{figure}
  \centering
  \includegraphics[width=\textwidth]{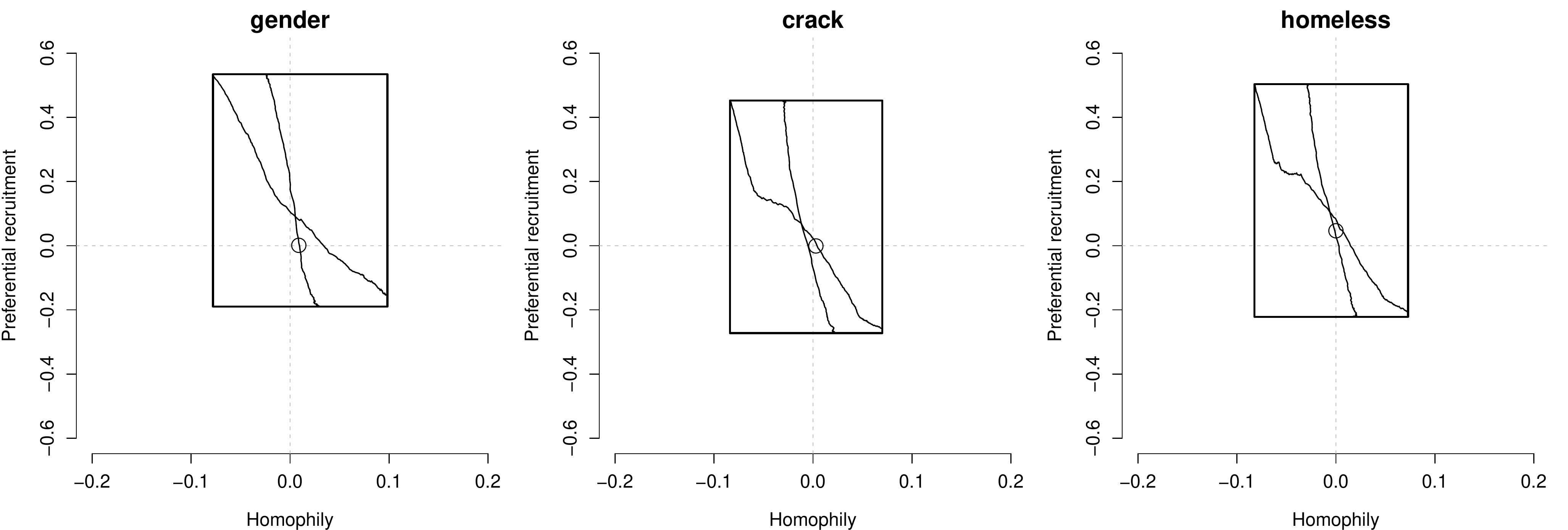}
  \caption{Bounds and true values for homophily and preferential recruitment with respect to gender, crack use, and homelessness in the RDS-net study.  For each trait, the box corresponds to the identification rectangle obtained by the Cartesian product of $I_h$ and $I_p$ in Table \ref{tab:results}.  The traces show the path of $(h,p)$ values visited by the optimization algorithm to find the extrema of $h$ and $p$.  The true value of $(h,p)$, calculated from the known augmented subgraph $G_{SU}$, is given by a circle.}
  \label{fig:results}
\end{figure}


\section{Discussion}

Researchers have devoted a great deal of attention to the influence of homophily and preferential recruitment on population-level estimates from RDS studies \citep{Gile2010Respondent,Tomas2011Effect,liu2012assessment,Rudolph2013Importance,lu2013linked,verdery2015respondent,rocha2015respondent}.  Under the assumptions articulated in this paper, neither of these sources of dependence can be calculated precisely from the observed data alone.  Consequently, there is reason to be skeptical of claims that particular populations surveyed by RDS exhibit homophily \citep[e.g.][]{Gwadz2011Effect,Simpson2014Birds,rudolph2011subpopulations,wejnert2012estimating,Rudolph2014Evaluating}, or that a particular RDS study suffers from preferential recruitment \citep[e.g.]{Iguchi2009Simultaneous,Yamanis2013Empirical,Young2014Spatial}.  It may be the case that homophily and preferential recruitment can induce bias in certain estimators of population quantities, but precisely diagnosing these pathologies is not possible in most RDS studies.

Although it may be disappointing that subgraph homophily and preferential recruitment are usually not point identified in RDS studies, we can still draw credible inferences about these parameters.  For example, the identification rectangles for gender, crack use, and homelessness in the RDS-net study are considerably smaller in area than the outcome space $[-1,1]^2$.  Under some circumstances, it may be possible to deduce that homophily or preferential recruitment is strictly positive or negative in the augmented subgraph, even without exact knowledge of that subgraph.  Informally, the more $G_{SU}$ resembles $G_R$, the narrower the identification region for $(h,p)$ will be.

The identification bounds proposed in this paper depend on three fundamental assumptions: the network exists, subjects are recruited across its edges, and nobody can be recruited more than once.  When these assumptions are met, the structure of data from RDS studies allows computation of credible bounds for $h$ and $p$.  However, the observed data also impose strict limits on the precision of these estimates: the bounds are often wide in practice.  Stronger assumptions about the topology of the network and dynamics of the recruitment process may yield narrower bounds, or point identification, at the cost of decreased credibility \citep{Manski2003Partial}.

In some circumstances, Assumptions \ref{assump:net}-\ref{assump:nomult} may not be reasonable.  For example, \citet{Scott2008They} describes a study in which subjects reported selling their coupons instead of recruiting among their social contacts.  In post-recruitment follow-up interviews, 119 subjects in the RDS-net study reported having given a coupon to someone other than the person who redeemed it.  By defining $G_R$ as the coupon redemption graph and $G_{SU}$ as the network of possible coupon redemptions, we have tried to mitigate violations of Assumption \ref{assump:recruit}.  Even if subjects truly recruit only their yet-unrecruited neighbors in an idealized social network, they may misreport their degrees in the network \citep{mccarty2001comparing,Salganik2006Variance,bell2007partner}.  Researchers may be able to improve the reliability of degree reports by administering a follow-up questionnaire to subjects about their recruitment behavior \citep{deMello2008Assessment,Yamanis2013Empirical,Gile2015Diagnostics}, or by statistical estimation of degree from enhanced survey instruments \citep{Zheng2006How,McCormick2010How,salganik2011game}. Researchers can assess the sensitivity of the proposed bounds to misreported degree by perturbing reported degrees according to a probability model.  For example, researchers could posit a sampling distribution for subjects' true degrees in $G$, and assess the variability of the identification bounds for $h$ and $p$ by marginalizing (or maximizing) over imputed degrees.


\vspace{0.5cm}

\noindent\textbf{Acknowledgements:} FWC was supported by NIH Grant KL2 TR000140, NIMH grant P30MH062294, the Yale Center for Clinical Investigation, and the Yale Center for Interdisciplinary Research on AIDS.  LZ was supported by a fellowship from the Yale World Scholars Program sponsored by the China Scholarship Council.
RDS-net was funded by NIH/NIDA grant 5R01DA031594-03 to Jianghong Li. 
We acknowledge the Yale University Biomedical High Performance Computing Center for computing support, funded by NIH grants RR19895 and RR029676-01.
We thank 
Gayatri Moorthi, 
Heather Mosher,
Greg Palmer, 
Eduardo Robles, 
Mark Romano, Jason Weiss, and the staff at the Institute for Community Research for their work collecting and preparing the RDS-net data.  
We are grateful to
Jacob Fisher,
Krista Gile, 
Mark Handcock, 
Robert Heimer, 
Edward H. Kaplan,
Lilla Orr,
Jiacheng Wu,
and
Alexei Zelenev
for helpful conversations and comments on the manuscript.

\section*{Appendix 1: Proofs}


\begin{proof}[Proof of Proposition \ref{prop:GS}] 
  Call a recruitment-induced subgraph $G_S=(V_S,E_S)$ compatible with the observed data if $V_S=V_R$, $\{i,j\}\in E_R$ implies $\{i,j\}\in E_S$, and $\sum_{j\neq i} \indicator{\{i,j\}\in E_S} \le d_i$ for each $i\in V_R$.    Call an augmented recruitment-induced subgraph $G_{SU}=(V_{SU},E_{SU})$ compatible with the observed data if conditions 1, 2, 4, and 5 of Definition \ref{defn:compatibility} hold.  Suppose $i\in V_R$ has $d_i^r<d_i$ and $j\in V_R$ has $d_j^r<d_j$.  Let $G_{SU}^1=(V_{SU}^1,E_{SU}^1)$ be any compatible subgraph with $\{i,u\}\in E_{SU}$, $\{j,u\}\in E_{SU}$, where $u\notin V_R$ is an unsampled vertex.  Let $G_{SU}^2=(V_{SU}^2,E_{SU}^2)$ be identical to $G_{SU}^1$ except that $\{i,j\}\in E_{SU}^2$, so neither $i$ nor $j$ is connected to $u$.  If in the resulting subgraph $u$ has no neighbors in $V_R$, i.e. there does not exist $k\in V_R$ such that $\{k,u\}\in E_{SU}^2$, then remove $u$ from $V_{SU}^2$.  
  Let $G_S$ be the recruitment-induced subgraph obtained by removing any unsampled vertices (and edges connected to them) from $G_{SU}$.  Clearly $G_S$ is compatible with the observed data, and $G_{SU}^2$ is compatible under conditions 1,2,4, and 5 of Definition \ref{defn:compatibility}. Since there exist at least two compatible recruitment-induced subgraphs and at least two compatible augmented recruitment-induced subgraphs, neither $G_S$ nor $G_{SU}$ are uniquely identified.
\end{proof}


\begin{proof}[Proof of Proposition \ref{prop:homophilyid}]
  Suppose the observed RDS data are $G_R$, $\Z_R$, $\T_R$, $\bd_R$, and there exist distinct $i\in V_R$ and $j\in V_R$ with $d_i^r<d_i$, $d_j^r<d_j$, and $Z_i\neq Z_j$.  Without loss of generality, suppose $Z_i=0$ and $Z_j=1$.  We will exhibit $(G_{SU}^1,\Z_{SU}^1)\in \mathcal{C}(G_R,\bd_R,\Z_R)$ and $(G_{SU}^2,\Z_{SU}^2)\in\mathcal{C}(G_R,\bd_R,\Z_R)$ such that $h(G_{SU}^1,G_R,\Z_{SU}^1) \neq h(G_{SU}^2,G_R,\Z_{SU}^2)$.  Let $(G_{SU}^1,\Z_{SU}^1)$ be any compatible subgraph and trait set with the property that $\{i,u_1\}\in E_{SU}$ and $\{j,u_2\}\in E_{SU}$, where $u_1$ and $u_2$ are unsampled vertices with $Z_{u_1}=0$ and $Z_{u_2}=1$.  Let $(G_{SU}^2,\Z_{SU}^2)$ be identical to $(G_{SU}^1,\Z_{SU}^1)$ except that the edges connecting $i$ and $j$ to $u_1$ and $u_2$ are swapped: $\{i,u_1\}\notin E_{SU}$,$\{j,u_2\}\notin E_{SU}$, and $\{i,u_2\}\in E_{SU}^2$ and $\{j,u_1\}\in E_{SU}^2$.  Clearly we have $(G_{SU}^2,\Z_{SU}^2) \in \mathcal{C}(G_R,\bd_R,\Z_R)$.  Note that $\bar{A}_{SU}$, $\bar{Z}_{SU}$, $\sigma(\A_{SU})$, and $\sigma(\Z_{SU})$ are the same under both $(G_{SU}^1,\Z_{SU}^1)$ and $(G_{SU}^2,\Z_{SU}^2)$.  Let $h_1=h(G_{SU}^1,G_R,\Z_{SU}^1)$ and $h_2=h(G_{SU}^2,G_R,\Z_{SU}^2)$ be the calculated values of homophily.  We compute the difference 
\begin{equation}
 \begin{split} 
 h_1 - h_2 &= \frac{ 2(1-\bar{A})(1-\bar{Z}) + 2(0-\bar{A})(0-\bar{Z}) - 2(1-\bar{A})(0-\bar{Z}) -2(0-\bar{A})(1-\bar{Z}) }{\left(\binom{|V_R|}{2} + |V_R||U| \right)\sigma(\A) \sigma(\Z)} \\
           &= \frac{2}{ \left(\binom{|V_R|}{2} + |V_R||U| \right) \sigma(\A) \sigma(\Z)}  > 0 .
 \end{split}
\end{equation}
Since this quantity is always non-zero, homophily is not point identified.
\end{proof}


\begin{proof}[Proof of Proposition \ref{prop:prefrecid}]
  Again suppose the observed RDS data are $G_R$, $\Z_R$, $\T_R$, $\bd_R$, and there exists $i\in V_R$ such that $d_i^r<d_i$ and $i$ recruited $k\in V_R$, $k\neq i$.  Without loss of generality, suppose $Z_i=1$.  Let $(G_{SU}^1,\Z_{SU}^1)$ be any compatible subgraph and trait set with the property that one edge connects $i$ to an unsampled vertex $u$, where $u$ has no other neighbors in $V_R$, and $Z_u=1$.  Let $(G_{SU}^2,\Z_{SU}^2)$ be identical to $(G_{SU}^1,\Z_{SU}^1)$ except that $Z_u=0$.  Recall that $|S_{r_j}(t_j)|$ is the number of susceptible vertices connected to the recruiter $r_j$ of $j$ in $G_{SU}^1$ or $G_{SU}^2$.  The difference is 
\[ p(G_{SU}^1,G_R,\mathbf{t}_R,\Z_{SU}^1) - p(G_{SU}^2,G_R,\mathbf{t}_R,\Z_{SU}^2) = \frac{1}{n-|M|} \frac{1}{|S_i(t_k)|} > 0. \]
Therefore $p$ is not point identified.
\end{proof}


\begin{proof}[Proof of Proposition \ref{prop:simanneal}]

Let $J(h,p) = 1/(1+\epsilon+ h)$ for $0<\epsilon<1$ and let $\mathcal{M}$ be the set of $(G_{SU},\Z_{SU})$ that achieve the global maximum of $J$ on $\mathcal{C}(G_R,\bd_R,\Z_R)$.  Let the cooling schedule be given by 
\[ T_t = \frac{1}{\epsilon\log(t)} . \]
Following \citet{Hajek1988Cooling}, we say that a state $(G_{SU},\Z_{SU})\in \mathcal{C}(G_R,\bd_R,\Z_R)$ communicates with $\mathcal{M}$ at depth $D$ if there exists a path in $\mathcal{C}(G_R,\bd_R,\Z_R)$ that starts at $(G_{SU},\Z_{SU})$ and ends at an element of $\mathcal{M}$ such that the least value of $J$ along the path is $J\big(h(G_{SU},G_R,\Z_{SU}),p(G_{SU},G_R,\mathbf{t}_R,\Z_{SU})\big)-D$.  Let $D^*$ be the smallest number such that every $(G_{SU},\Z_{SU})\in \mathcal{C}(G_R,\bd_R,\Z_R)$ communicates with $\mathcal{M}$ at depth $D^*$.  Theorem 1 of \citet{Hajek1988Cooling} states that if $T_t \to 0$ and $\sum_{t=1}^\infty \exp[ -D^*/T_t ]$ diverges, then the sequence $(G_{SU},\Z_{SU})_t$ converges in probability to an element of $\mathcal{M}$.  

First, note that since $J(h,p)>0$ for all $h$, $D^*$ is bounded above by the maximum of $J$ on $\mathcal{C}(G_R,\bd_R,\Z_R)$, and so   
  \begin{equation}
    \begin{split}
      D^* &\le \max_{(G_{SU},\Z_{SU})\in \mathcal{C}(G_R,\bd_R,\Z_R)} J\big(h(G_{SU},G_R,\Z_{SU}),p(G_{SU},G_R,\mathbf{t}_R,\Z_{SU})\big) \\
      &\le \max_{(h,p)\in [-1,1]^2} J(h,p) \\
      &= \max_{(h,p)\in [-1,1]^2} 1/(1+\epsilon+h) \\
      &= 1/\epsilon .
  \end{split}
  \end{equation}
  Now examining the divergence criterion, 
  \begin{equation}
    \begin{split}
      \sum_{t=1}^\infty \exp[ -D^*/T_t ] &= \sum_{t=1}^\infty \exp\left[ - D^* \epsilon \log(t) \right] \\
                                         &= \sum_{t=1}^\infty \frac{1}{t^{D^*\epsilon}} \\
                                         &\ge \sum_{t=1}^\infty \frac{1}{t} = \infty 
    \end{split}
  \end{equation}
  where the inequality is a consequence of $D^*\epsilon \le 1$.  Therefore $\lim_{t\to\infty} \Pr\big( (G_{SU},\Z_{SU})_t \in \mathcal{M} \big) = 1$, as claimed.
\end{proof}


\section*{Appendix 2: Sampling $(G_{SU},\Z_{SU})$}

Suppose $(G_{SU},\Z_{SU})\in \mathcal{C}(G_R,\bd_R,\Z_R)$ is a compatible augmented subgraph and trait set, and we wish to propose another compatible pair $(G_{SU}^*,\Z_{SU}^*)\in \mathcal{C}(G_R,\bd_R,\Z_R)$. We outline two proposal mechanisms. The first removes or adds an edge in $G_{SU}$.  If necessary, a new unsampled vertex $u$ is invented, and assigned a trait value $Z_u$.  Let $U=\{u\in V_{SU}:\ u\notin V_R\}$ be the set of unsampled vertices.  
Furthermore, let $U_{-k}=\{u\in V_{SU}\setminus V_R:\ \{k,u\}\notin E_{SU}\}$ be the set of unsampled vertices in $U$ that are \emph{not} connected to $k\in V_R$.
\begin{algorithmic}[1]
\STATE Let $G_{SU}^*=G_{SU}$ and $\Z_{SU}^*=\Z_{SU}$ 
\STATE Randomly choose $i\in V_R$ and $j\in V_{SU}$ with $i\neq j$.
\IF{ $\{i,j\}\in E_{SU}$ and $\{i,j\}\notin E_R$ } 
  \STATE Remove $\{i,j\}$ from $E_{SU}^*$
  \STATE $B \sim \text{Bernoulli}(1/2)$
  \IF{$B<0.5$ and $U_{-i}\neq\emptyset$}
    \STATE Randomly choose $u\in U_{-i}$
  \ELSE 
    \STATE Add a new vertex $u$ to $V_{SU}^*$
    \STATE Randomly choose a trait $Z_u^*\in\{0,1\}$ 
  \ENDIF
  \STATE Add $\{i,u\}$ to $E_{SU}^*$
  \IF{ $j\in V_R$ } 
    \STATE $B \sim \text{Bernoulli}(1/2)$
    \IF{$B<0.5$ and $U_{-j}\neq\emptyset$}
      \STATE Randomly choose $u\in U_{-j}$
    \ELSE 
      \STATE Add a new vertex $u$ to $V_{SU}^*$
      \STATE Randomly choose a trait $Z_u^*\in\{0,1\}$ 
    \ENDIF
  \ENDIF
  \STATE Add $\{j,u\}$ to $E_{SU}^*$
  \ELSIF{$\{i,j\}\notin E_{SU}$ and $\exists u_1,u_2\in U:\ \{i,u_1\}\in E_{SU}$ and $\{j,u_2\}\in E_{SU}$} 
    \STATE Remove $\{i,u_1\}$ and $\{j,u_2\}$ from $E_{SU}^*$
    \STATE Add $\{i,j\}$ to $E_{SU}^*$  
\ENDIF
\STATE Remove any isolated vertices from $V_{SU}^*$
\end{algorithmic}
The space $\mathcal{C}(G_R,\bd_R,\Z_R)$ is connected via proposals of this type \citep[see][for explanation]{Crawford2014Graphical}.  The second proposal mechanism accelerates exploration of $\mathcal{C}(G_R,\bd_R,\Z_R)$ by switching the trait of an unsampled vertex:  
\begin{algorithmic}[1]
  \STATE Choose $u\in \{u\in V_{SU}:\ u\notin V_R \}$.
  \STATE Set $Z_u^* = 1-Z_u$.
\end{algorithmic}
Together, these proposal mechanisms result in a well-mixing sequence $(G_{SU},\Z_{SU})_t$.


\bibliographystyle{spbasic}
\bibliography{rds}

\end{document}